\documentclass[12pt]{article}

\usepackage[margin=1in]{geometry}
\usepackage{CJKutf8, amsmath, amssymb, amsthm, authblk, hyperref, orcidlink}
\usepackage[basic]{complexity}

\newtheorem{lemma}{Lemma}
\newtheorem{theorem}{Theorem}
\newtheorem{corollary}{Corollary}
\newtheorem{conjecture}{Conjecture}

\theoremstyle{definition}
\newtheorem{definition}{Definition}
\newtheorem{question}{Question}

\DeclareMathOperator{\tr}{tr}
\DeclareMathOperator*{\e}{\mathbb E}

\usepackage[style=trad-unsrt, citestyle=numeric-comp, giveninits=true, doi=false, url=false, eprint=false, isbn=false]{biblatex}

\begin{document}

\title{Long-time properties of generic Floquet systems are approximately periodic with the driving period}

\author{Yichen Huang (黄溢辰)\orcidlink{0000-0002-8496-9251}\thanks{yichenhuang@fas.harvard.edu}}
\affil{Department of Physics, Harvard University, Cambridge, Massachusetts 02138, USA}

\begin{CJK}{UTF8}{gbsn}

\maketitle

\end{CJK}

\begin{abstract}

A Floquet quantum system is governed by a Hamiltonian that is periodic in time. Consider the space of piecewise time-independent Floquet systems with (geometrically) local interactions. We prove that for all but a measure zero set of systems in this space, starting from a random product state, many properties (including expectation values of observables and the entanglement entropy of a macroscopically large subsystem) at long times are approximately periodic with the same period as the Hamiltonian. Thus, in almost every Floquet system of arbitrarily large but finite size, discrete time-crystalline behavior does not persist to strictly infinite time.

\end{abstract}

\section{Introduction}

Floquet quantum systems are periodically driven systems whose Hamiltonian is a periodic function of time \cite{MS17, HRRS20, HD22}. They are of high current interest \cite{RSS12, LDM14PRL, LDM14PRE, DR14, PCPA15, PPHA15} partly because they are believed to host certain dynamical phases of matter \cite{KLMS16, vKS16, EBN16, YPPV17, ZHK+17, CCL+17, SZ17, EMNY20, MIQ+22, ZLM+23, KMS19} that are forbidden in systems governed by a time-independent Hamiltonian.

Let $H(t)$ be a time-dependent Hamiltonian such that
\begin{equation}
H(t+1)=H(t),\quad\forall t\in\mathbb R,
\end{equation}
where we assume without loss of generality that the driving period is $1$. Let $\mathcal T$ be the time-ordering operator and
\begin{equation} \label{eq:teo}
U(t_0,t_1):=\mathcal Te^{-i\int_{t_0}^{t_1}H(t)\,\mathrm dt},\quad t_0\le t_1
\end{equation}
be the unitary time-evolution operator from time $t_0$ to $t_1$. Define the Floquet operator
\begin{equation} \label{eq:flqo}
U_F:=U(0,1)
\end{equation}
as the time-evolution operator in one period. At time $t=0$, we initialize the system in a state $|\psi(0)\rangle$. The state at time $t\ge0$ is
\begin{equation}
|\psi(t)\rangle=U(0,t)|\psi(0)\rangle.
\end{equation}

While the Hamiltonian $H(t)$ is periodic in time, the trajectory $|\psi(t)\rangle_{t\in\mathbb R}$ in the Hilbert space is usually far from periodic. However, one might expect that
\begin{conjecture} [vague statement] \label{conj}
Long-time properties of $|\psi(t)\rangle$ are periodic with period $1$.
\end{conjecture}

If the property under consideration is the expectation value of an observable, Refs.~\cite{RSS12, LDM14PRL, LDM14PRE} give strong heuristic arguments for this conjecture. However, some important aspects of the conjecture remain to be understood.

Conjecture \ref{conj} does not hold for all Floquet systems. A simple counterexample is a non-interacting multi-spin system where $H(t)$ for any $t\in\mathbb R$ is a sum of on-site terms. If the initial state $|\psi(0)\rangle$ is a product state, so is $|\psi(t)\rangle$ for any $t\in\mathbb R$. Each spin evolves independently so that Conjecture \ref{conj} fails for almost every non-interacting system.

\begin{question} \label{q:1}
If Conjecture \ref{conj} holds for almost all Floquet systems, how to quantify ``almost all?''
\end{question}

We are particularly interested in initial states of low complexity, especially product states, partly because they are easier to realize experimentally. 

\begin{question} \label{q:2}
Even in non-integrable Floquet systems, there might be a small number of ``bad'' initial states for which Conjecture \ref{conj} fails. In the ensemble of product states, can we upper bound the fraction of such bad states? Does the fraction vanish in the thermodynamic limit?
\end{question}

\begin{question} \label{q:3}
In finite-size systems, properties of $|\psi(t)\rangle$ are not exactly but only approximately periodic even at long times. Can we upper bound the approximation error and how does it vanish in the thermodynamic limit?
\end{question}

\begin{question} \label{q:4}
What properties other than expectation values of observables does Conjecture \ref{conj} hold for? For example, is the entanglement entropy of $|\psi(t)\rangle$ (approximately) periodic with period $1$ at long times?
\end{question}

In this paper, we first formulate Conjecture \ref{conj} as a precise mathematical statement. Then, we prove the statement. In this process, we answer all questions above rigorously.

\paragraph{Results (informal statement).}Consider the space of piecewise time-independent Floquet systems with (geometrically) local interactions. We prove that for all but a measure zero set of systems in this space, starting from a random product state, with overwhelming probability many properties (including expectation values of observables and the entanglement entropy of a macroscopically large subsystem) at long times are approximately periodic with the same period as the Hamiltonian.

\section{Discrete time crystals}

A characteristic feature of a quantum time crystal is that its dynamics breaks the time-translation symmetry of the Hamiltonian \cite{Wil12}. However, (spontaneous) time-translation symmetry breaking can be defined in different ways. Sometimes additional requirements are imposed to exclude certain unnatural systems as time crystals. Thus, there are multiple definitions of time crystals in the literature (see Ref.~\cite{KMS19} for a review of some of the definitions), but none of them is universally accepted as the most standard. Partly for this reason, there is much debate \cite{Bru13Q, Bru13S, Wil13} in the community on what is a time crystal. See Ref.~\cite{ZLM+23} for a careful discussion on defining time crystals.

By definition, the Hamiltonian of a Floquet system (with period $1$) has the time-translation symmetry $t\to t+\mathbb Z$. A Floquet system is a time crystal if the dynamics of some property does not even approximately have this symmetry but approximately has a smaller symmetry $t\to t+m\mathbb Z$ for some integer $m>1$. Since the symmetry group $\mathbb Z$ of the Hamiltonian is discrete, Floquet time crystals \cite{EBN16} are also called discrete time crystals (DTC) \cite{YPPV17, EMNY20, ZLM+23}.

The result stated in the last paragraph of the introduction implies that in almost every piecewise time-independent Floquet system of arbitrarily large but finite size, DTC behavior does not persist to strictly infinite time or the lifetime of the DTC order is at most finite.\footnote{Equivalently, in piecewise time-independent Floquet systems of arbitrarily large but finite size, realizing DTC with strictly infinite lifetime requires fine-tuning the parameters of the Hamiltonian to infinite precision.} Ample numerical evidence and heuristic arguments in the literature \cite{vKS16, EBN16, YPPV17, LM17, ZLM+23, KMS19} suggest that in a particular class of models, the DTC lifetime is exponential in the system size. Our contribution is to prove a precise mathematical statement of the non-infinite lifetime of the DTC order in a general setting. To our knowledge, this is the first rigorous result on the lifetime of DTC.

\section{Basic results}

In this section, we present the basic version of our result. We specify the models and initial states that our results apply to. We also define the concept of approximately periodic functions in order to quantify the accuracy with which Conjecture \ref{conj} holds. In the next section, we will extend the main result of this section in many directions.

In a system of $N$ qubits (spin-$1/2$'s) labeled by $1,2,\ldots,N$, let
\begin{equation}
\sigma^x_l=\begin{pmatrix}0&1\\1&0\end{pmatrix},\quad
\sigma^y_l=\begin{pmatrix}0&-i\\i&0\end{pmatrix},\quad
\sigma^z_l=\begin{pmatrix}1&0\\0&-1\end{pmatrix}
\end{equation}
be the Pauli matrices for qubit $l$. Let
\begin{equation}
\mathbf h:=(h_l^x,h_l^z)_{1\le l\le N}\in\mathbb R^{2N},\quad\mathbf J:=(J_l)_{1\le l\le N-1}\in\mathbb R^{N-1},
\end{equation}
and consider the piecewise time-independent Floquet system (with period $1$)
\begin{equation} \label{eq:modelb}
H(t)=
\begin{cases}
H'_1:=\sum_{l=1}^Nh_l^z\sigma_l^z+\sum_{l=1}^{N-1}J_l\sigma_l^z\sigma_{l+1}^z,&0\le t<1/2\\
H'_2:=\sum_{l=1}^Nh_l^x\sigma_l^x,&1/2\le t<1
\end{cases}.
\end{equation}
This model is in different dynamical phases for different values of $(\mathbf h,\mathbf J)$ \cite{YPPV17}.

\begin{definition} [Haar-random product state \cite{LQ19, Hua21ISIT, Hua22TIT, Hua23}] \label{def:haar}
\begin{equation} 
|\phi\rangle=\bigotimes_{l=1}^N|\phi_l\rangle
\end{equation}
is a Haar-random product state if each $|\phi_l\rangle$ (the state of qubit $l$) is chosen independently and uniformly at random with respect to the Haar measure.
\end{definition}

\begin{definition} [approximately periodic function] \label{def:apf}
A function $f:\mathbb R\to\mathbb R$ is $\epsilon$-periodic (with period $1$) if there exists another function $\bar f:[0,1)\to\mathbb R$ such that
\begin{equation} \label{eq:apf}
\liminf_{\mathbb N\ni M\to\infty}\frac1M\left|\left\{m\in\{0,1,2,\ldots,M-1\}:\int_0^1|f(m+x)-\bar f(x)|^2\,\mathrm dx\le\epsilon\right\}\right|\ge1-\epsilon.
\end{equation}
\end{definition}

(\ref{eq:apf}) means the following. $\int_0^1|f(m+x)-\bar f(x)|^2\,\mathrm dx$ is the squared $L^2$ distance between the restriction of $f$ to one period $[m,m+1)$ and a particular function $\bar f$. We say that $[m,m+1)$ is a good period if the squared distance is $\le\epsilon$. Otherwise, it is a bad period. We count the number of good periods and require that the fraction of good periods in the limit of large $M$ is $\ge1-\epsilon$. In other words, we allow a small fraction of bad periods in which $f$ deviates significantly from periodic. In all good periods, $f$ must be approximately periodic with error $\le\epsilon$. We use the limit inferior in (\ref{eq:apf}) because the limit may not exist.

Let $A$ be an observable (Hermitian operator). Assume without loss of generality that its operator norm is $\|A\|=1$.

We use standard asymptotic notation. Let $g_1,g_2:\mathbb R^+\to\mathbb R^+$ be two functions. One writes $g_1(x)=O(g_2(x))$ if and only if there exist constants $C,x_0>0$ such that $g_1(x)\le Cg_2(x)$ for all $x>x_0$; $g_1(x)=\Omega(g_2(x))$ if and only if there exist constants $C,x_0>0$ such that $g_1(x)\ge Cg_2(x)$ for all $x>x_0$; $g_1(x)=\Theta(g_2(x))$ if and only if there exist constants $C_1,C_2,x_0>0$ such that $C_1g_2(x)\le g_1(x)\le C_2g_2(x)$ for all $x>x_0$; $g_1(x)=o(g_2(x))$ if and only if for any constant $C>0$ there exists a constant $x_0>0$ such that $g_1(x)<Cg_2(x)$ for all $x>x_0$.

\begin{theorem} \label{t:b}
Let the initial state $|\psi(0)\rangle$ be a Haar-random product state (Definition \ref{def:haar}). For the model (\ref{eq:modelb}), the set
\begin{equation} \label{eq:b}
\mathbb R^{3N-1}\setminus\left\{(\mathbf h,\mathbf J)\in\mathbb R^{3N-1}:\Pr_{|\psi(0)\rangle}\big(\langle\psi(t)|A|\psi(t)\rangle~\textnormal{is}~e^{-\Omega(N)}\textnormal{-periodic}\big)=1-e^{-\Omega(N)}\right\}
\end{equation}
has Lebesgue measure zero.
\end{theorem}

Equation (\ref{eq:modelb}) is an ensemble of Floquet systems parameterized by $(\mathbf h,\mathbf J)\in\mathbb R^{3N-1}$. Theorem \ref{t:b} says that for all but a measure zero set of Floquet systems in this ensemble, Conjecture \ref{conj} holds with exponential accuracy: $\langle\psi(t)|A|\psi(t)\rangle$ is $e^{-\Omega(N)}$-periodic for all but an exponentially small fraction of initial product states.

Theorem \ref{t:b} answers Questions \ref{q:1}, \ref{q:2}, \ref{q:3} for the model (\ref{eq:modelb}). Specifically, ``almost all'' in Question \ref{q:1} means ``full Lebesgue measure.'' In the ensemble of product states, the fraction of bad initial states defined in Question \ref{q:2} is exponentially small in the system size. The approximation error in Question \ref{q:3} is also exponentially small in the system size.

Due to the limit $M\to\infty$ in (\ref{eq:apf}), whether a function $f$ is approximately periodic is determined by the values $f(x)$ for sufficiently large $x$. Thus, ``$\langle\psi(t)|A|\psi(t)\rangle$ is $e^{-\Omega(N)}$-periodic'' in (\ref{eq:b}) is a statement at the longest time scales. Theorem \ref{t:b} is silent on the dynamic properties of the model (\ref{eq:modelb}) at short and intermediate times.

\section{Full results}

In this section, we extend Theorem \ref{t:b} in many directions.

A Floquet system $H(t)$ (with period $1$) is mathematically a function from $[0,1)$ to Hermitian matrices. To extend Theorem \ref{t:b} to ensembles of general Floquet systems, we need to define ``measure zero''  in the space of functions from $[0,1)$ to Hermitian matrices. However, defining ``measure zero'' in such infinite-dimensional function spaces is a well-known mathematical challenge, although progress has been made (see the review article \cite{OY05}).

In order to be less mathematically involved, in the remainder of this section we only consider piecewise time-independent Floquet systems (of finite size), which are parameterized by a finite number of real parameters. Thus, we avoid infinite-dimensional function spaces and can use the Lebesgue measure. Note that many prototypical Floquet systems (including the models of many-body localization and DTC in Refs.~\cite{PCPA15, PPHA15, vKS16, KLMS16, EBN16, YPPV17}) are piecewise time independent.

Specifically, we define ensembles of piecewise time-independent Floquet systems with nearest-neighbor interactions in a chain of $N$ qubits. We prove that in each ensemble, Conjecture \ref{conj} holds with exponential accuracy for all but a Lebesgue measure zero set of systems. Similar results can be proved in a similar way for other types of piecewise time-independent Floquet systems including qudit systems with short-range interactions in higher spatial dimensions or even with non-local interactions.

Let $n$ be a positive integer and
\begin{gather}
T\in\{(t_0,t_1,t_2,\ldots,t_{n-1},t_n):0=t_0<t_1<t_2<\cdots<t_{n-1}<t_n=1\},\\
\alpha:=(\alpha_j^u)_{1\le j\le n}^{u\in\{x,y,z\}}\in\{0,1\}^{3n},\quad\gamma:=(\gamma_j^{u,v})_{1\le j\le n}^{u,v\in\{x,y,z\}}\in\{0,1\}^{9n}.
\end{gather}
Each tuple $(n,T,\alpha,\gamma)$ defines an ensemble of piecewise time-independent Floquet systems (with period $1$)
\begin{equation} \label{eq:f}
H_{n,T,\alpha,\gamma}(t)=H_j:=\sum_{l=1}^N\sum_{u\in\{x,y,z\}}\alpha_j^u h_{j,l}^u\sigma_l^u+\sum_{l=1}^{N-1}\sum_{u,v\in\{x,y,z\}}\gamma_j^{u,v} J_{j,l}^{u,v}\sigma_l^u\sigma_{l+1}^v,\quad t_{j-1}\le t<t_j,
\end{equation}
where
\begin{equation}
\mathbf h:=(h_{j,l}^u)_{1\le j\le n;1\le l\le N}^{u\in\{x,y,z\}}\in\mathbb R^{3nN},\quad\mathbf J:=(J_{j,l}^{u,v})_{1\le j\le n;1\le l\le N-1}^{u,v\in\{x,y,z\}}\in\mathbb R^{9n(N-1)}.
\end{equation}
In words, $n$ is the number of pieces in one period of the piecewise time-independent Floquet system. $t_0,t_1,\ldots,t_n$ are the times when the Hamiltonian can change abruptly. $\alpha_j^u$ ($\gamma_j^{u,v}$) indicates whether the $j$th piece of the Hamiltonian has a particular type of on-site (nearest-neighbor) terms. $\mathbf h$ and $\mathbf J$ give the coefficients of the terms.

\begin{theorem} \label{t:f}
Let the initial state $|\psi(0)\rangle$ be a Haar-random product state. For any tuple $(n,T,\alpha,\gamma)$ such that $\alpha_1^x=\alpha_1^z=\gamma_1^{zz}=1$, the set
\begin{equation}
\mathbb R^{12nN-9n}\setminus\left\{(\mathbf h,\mathbf J)\in\mathbb R^{12nN-9n}:\Pr_{|\psi(0)\rangle}\big(\langle\psi(t)|A|\psi(t)\rangle~\textnormal{is}~e^{-\Omega(N)}\textnormal{-periodic}\big)=1-e^{-\Omega(N)}\right\}
\end{equation}
has Lebesgue measure zero.
\end{theorem}

The condition $\alpha_1^x=\alpha_1^z=\gamma_1^{zz}=1$ in Theorem \ref{t:f} is not absolutely necessary if some other elements of $\alpha$ or $\gamma$ are $1$. For example, Theorem \ref{t:b} corresponds to the case that $\alpha_1^z=\gamma_1^{zz}=\alpha_2^x=1$ while all other elements of $\alpha,\gamma$ are $0$.

We also prove the approximate periodicity of reduced density matrices. Let $\mathbb D$ be the set of density matrices (positive semi-definite matrices with unit trace) and $\|\cdot\|_1$ denote the trace norm.

\begin{definition} [approximately periodic matrix function] \label{def:apmf}
A matrix function $f:\mathbb R\to\mathbb D$ is $\epsilon$-periodic (with period $1$) if there exists another matrix function $\bar f:[0,1)\to\mathbb D$ such that
\begin{equation}
\liminf_{\mathbb N\ni M\to\infty}\frac1M\left|\left\{m\in\{0,1,2,\ldots,M-1\}:\int_0^1\|f(m+x)-\bar f(x)\|_1\,\mathrm dx\le\epsilon\right\}\right|\ge1-\epsilon.
\end{equation}
\end{definition}

Let $S$ be an arbitrary (not necessarily connected) subsystem of $L$ qubits and $\bar S$ be the complement of $S$ (rest of the system). Let $\psi(t)_S:=\tr_{\bar S}|\psi(t)\rangle\langle\psi(t)|$ be the reduced density matrix of $S$ at time $t$.

\begin{theorem} \label{t:r}
Let the initial state $|\psi(0)\rangle$ be a Haar-random product state. For $L\le0.29248N$ and any tuple $(n,T,\alpha,\gamma)$ such that $\alpha_1^x=\alpha_1^z=\gamma_1^{zz}=1$, the set
\begin{equation} \label{eq:r}
\mathbb R^{12nN-9n}\setminus\left\{(\mathbf h,\mathbf J)\in\mathbb R^{12nN-9n}:\Pr_{|\psi(0)\rangle}\big(\psi(t)_S~\textnormal{is}~e^{-\Omega(N)}\textnormal{-periodic}\big)=1-e^{-\Omega(N)}\right\}
\end{equation}
has Lebesgue measure zero.
\end{theorem}

\begin{corollary} \label{c:r}
For any function $f:\mathbb D\to\mathbb R$ such that
\begin{equation}
\|\rho_1-\rho_2\|_1=e^{-\Omega(N)}\implies|f(\rho_1)-f(\rho_2)|=e^{-\Omega(N)},
\end{equation}
Theorem \ref{t:r} remains valid upon replacing $\psi(t)_S$ in (\ref{eq:r}) by $f(\psi(t)_S)$. An example of such a function is the von Neumann entropy of $\psi(t)_S$ \cite{Fan73, Aud07} or the entanglement entropy of $|\psi(t)\rangle$ between $S$ and $\bar S$:
\begin{equation}
f(\psi(t)_S)=-\tr\big(\psi(t)_S\ln\psi(t)_S\big).
\end{equation}
\end{corollary}

Theorem \ref{t:r} and Corollary \ref{c:r} answer Question \ref{q:4}. In almost every Floquet system initialized in a Haar-random product state, Conjecture \ref{conj} holds with exponential accuracy for the reduced density matrix of any subsystem up to a particular size and for any sufficiently continuous function of this reduced density matrix. It is an open question to what extent the condition $L\le0.29248N$ in Theorem \ref{t:r} and Corollary \ref{c:r} can be relaxed.

\section{Stability}

In physics, stability is a concept of significant practical importance. Physical phenomena that are unstable with respect to perturbations occur only in fine-tuned systems and thus may be difficult to observe in practice.

To formally define stability in (finite-size) quantum systems, consider an ensemble of models parameterized by a vector $\mathbf x\in\mathbb R^\kappa$ of $\kappa$ real numbers, where $\mathbb R^\kappa$ is called the parameter space. For example, Eq.~(\ref{eq:modelb}) is an ensemble of Floquet systems parameterized by $\mathbf x=(\mathbf h,\mathbf J)$ with $\kappa=3N-1$. We say that a property is stable at $\mathbf x_0\in\mathbb R^\kappa$ with perturbation strength $r$ if for any $\mathbf x\in\mathbb R^\kappa$ whose distance from $\mathbf x_0$ is $\le r$, the model represented by $\mathbf x$ has the property. The larger $r$ is, the more stable the property is. $r=0^+$ (infinitesimal perturbation) corresponds to stability in the weakest sense.

If models with a particular property occupy only a measure zero subset $X$ of the parameter space $\mathbb R^\kappa$, then the property is unstable at any $\mathbf x_0\in X$ with respect to infinitesimal perturbations. 

\begin{definition} [longest time scale] \label{def:lts}
In finite-size systems, a dynamic behavior occurs up to the longest time scale if it persists to strictly infinite time.
\end{definition}

Theorems \ref{t:b}, \ref{t:f} immediately imply

\begin{corollary} [Instability of DTC at longest times] \label{c:DTC}
In the Floquet system (\ref{eq:modelb}) of arbitrarily large but finite size $N$, DTC behavior at the longest time scales is unstable at any $(\mathbf h,\mathbf J)$ with respect to infinitesimal periodic perturbations of period $1$. The same holds for any ensemble of Floquet systems considered in Theorem \ref{t:f}.
\end{corollary}

Note that all numerical evidence and heuristic arguments in the literature \cite{vKS16, EBN16, KMS19, ZLM+23} for the stability of DTC are for short and intermediate times (up to time exponential in the system size). They do not contradict Corollary \ref{c:DTC}.

Since the heuristic arguments \cite{vKS16, EBN16, KMS19, ZLM+23} for the stability of DTC (up to intermediate times) require many-body localization (MBL) in Floquet systems, in the remainder of this section we compare the stability of DTC and that of MBL (Table \ref{t}). When discussing the stability of a dynamic property in a system of size $N$, we must specify the time scale and perturbation strength (denoted by $r$ in the second paragraph of this section) under consideration.

\begin{table}
\caption{Stability of MBL versus that of DTC with respect to perturbations that preserve the time-translation symmetry of the unperturbed Hamiltonian. The second row applies to MBL of both time-independent Hamiltonians and Floquet systems. $N$ is the system size.}
\label{t}
\centering
\begin{tabular}{c|c|c|c}
\hline
perturbation & infinitesimal or & $1/\poly(N)$ & extensive (Definition \ref{def:extp}) \\
strength & sufficiently small & & \\
\hline
MBL & stable at all time scales & Question \ref{q:5} & controversial (Question \ref{q:6})\\
\hline
DTC & unstable at the longest & stable up to & believed to be stable up to \\
& time scales (Corollary \ref{c:DTC}) & time $\ge\poly(N)$ & time $e^{\Omega(N)}$ (Conjecture \ref{c:2}) \\
\hline
\end{tabular}
\end{table}

\paragraph{MBL of time-independent Hamiltonians.}When speaking of the stability of MBL in systems governed by time-independent Hamiltonians, people usually consider the longest time scale: For arbitrarily large but finite $N$, upon perturbation does localization persist to infinite time?\footnote{Dynamics that is localized up to time exponential in $N$ but delocalized afterwards may be called quasi-many-body localization \cite{YLC+16}.} Suppose that the unperturbed Hamiltonian $H_\textnormal{unp}$ is extensive.

\begin{definition} [extensive Hamiltonian]
In a quantum lattice system with $N$ sites, a Hamiltonian is extensive if it is a sum of $N$ terms such that
\begin{itemize}
\item for each lattice site, there are $\Theta(1)$ terms whose support is contained in a constant-radius neighborhood of the site (the support of an operator is the set of sites it acts non-trivially on);
\item the operator norm of each term is $\Theta(1)$.
\end{itemize}
\end{definition}

\begin{definition} [non-degenerate spectrum and non-degenerate gap \cite{Per84, HBZ19, HH23}] \label{def:ndg}
The spectrum of a Hamiltonian is non-degenerate if all eigenvalues are distinct. A non-degenerate spectrum $\{e_j\}$ has non-degenerate gaps if for any $j\neq k$,
\begin{equation} \label{eq:ndg}
e_j-e_k=e_{j'}-e_{k'}\implies(j=j')~\textnormal{and}~(k=k').
\end{equation}
\end{definition}

Assume that the spectrum of $H_\textnormal{unp}$ is non-degenerate and has non-degenerate gaps. This assumption is reasonable because the spectra of all but a measure zero set of (geometrically) local Hamiltonians are such \cite{Hua21PP}. For many initial states (including but not limited to Haar-random product states \cite{HH19}), reduced density matrices of not too large subsystems become almost time independent at long times \cite{Tas98, Rei08, LPSW09, Sho11, SF12}. This process is called equilibration \cite{GE16, WGRE19}. The equilibration time is the time required to reach equilibrium. An upper bound, which we denote by $T$, on the equilibration time was derived \cite{SF12}. $T$ is a function of the energy spectrum and is finite (in finite-size systems) if the spectrum has non-degenerate gaps. After the equilibration time, properties of the system are described \cite{Rei08, LPSW09} by the so-called diagonal ensemble \cite{RDO08}, which is a mixed state obtained by dephasing the initial state in the energy eigenbasis.

For a sufficiently small perturbation $H_\textnormal{per}$, the eigenvalues, eigenstates and hence the diagonal ensemble and the upper bound $T$ on the equilibration time of the perturbed Hamiltonian $H_\textnormal{unp}+H_\textnormal{per}$ are almost the same as those of $H_\textnormal{unp}$, respectively. Thus, at time $t>T$, the effect of $H_\textnormal{per}$ on the dynamics is small. For $0\le t\le T$, the effect is upper bounded by
\begin{equation} \label{eq:sp}
\|e^{-i(H_\textnormal{unp}+H_\textnormal{per})t}-e^{-iH_\textnormal{unp}t}\|\le\|H_\textnormal{per}\|t\le\|H_\textnormal{per}\|T\ll1
\end{equation}
if the perturbation strength $\|H_\textnormal{per}\|$ is $\ll1/T$. If the dynamics generated by $H_\textnormal{unp}$ is localized at all times, so is that by $H_\textnormal{unp}+H_\textnormal{per}$. Thus, MBL is stable at all time scales with respect to sufficiently small perturbations.

For slightly larger perturbations, it is an open question that
\begin{question} \label{q:5}
Is MBL stable at all time scales with respect to perturbations of strength $1/\poly(N)$?
\end{question}

In my personal opinion, a convincing answer to this question is an important step towards understanding the stability of MBL.

Arguably, the most realistic perturbation strength is extensive.

\begin{definition} [extensive perturbation] \label{def:extp}
In a quantum lattice system with $N$ sites, a perturbation is extensive if it is a sum of $N$ terms such that
\begin{itemize}
\item for each lattice site, there are $\Theta(1)$ terms whose support is contained in a constant-radius neighborhood of the site;
\item the operator norm of each term is upper bounded by a sufficiently small positive constant.
\end{itemize}
\end{definition}

\begin{question} \label{q:6}
Is MBL stable at all time scales with respect to extensive perturbations?
\end{question}

This is the most important open question about MBL. However, there is much debate on it in the community with refreshed understanding \cite{MCK+22} emerged in the last few years. The most rigorous evidence for a positive answer is Imbrie's proof in a random spin chain \cite{Imb16}. The proof assumes limited level attraction, but it is very difficult to prove this assumption. There are numerical evidence \cite{SBPV20, Sel22} in one spatial dimension and heuristic arguments \cite{DH17, DI17} in two and higher dimensions for a negative answer to Question \ref{q:6}. In general, numerical methods are limited to relatively small system sizes. They suffer from significant finite-size effects and are thus inconclusive \cite{PSS+20, ABD+21}.

\paragraph{Floquet MBL.}For the stability of MBL in Floquet systems, people also usually consider the longest time scale. Floquet MBL is stable at all time scales with respect to sufficiently small perturbations that preserve the time-translation symmetry of the unperturbed Hamiltonian. This result is obtained by extending the stability proof around (\ref{eq:sp}) for time-independent Hamiltonians to Floquet systems. Technical lemmas necessary for this extension are given in Appendix \ref{app}. References \cite{RSS12, LDM14PRL, LDM14PRE} extended the diagonal ensemble to Floquet systems.

The analogues of Questions \ref{q:5}, \ref{q:6} for Floquet MBL are important open questions.

\paragraph{DTC.}Corollary \ref{c:DTC} says that DTC at the longest time scales (Definition \ref{def:lts}) is unstable with respect to infinitesimal perturbations that preserve the time-translation symmetry of the unperturbed Hamiltonian. Therefore, when speaking of the stability of DTC, people usually consider short and intermediate times.

\begin{lemma} [Lemma 1 in Ref.~\cite{HC15}]
Let $G_1(t),G_2(t)$ be two time-dependent Hamiltonians such that
\begin{equation}
\|G_1(t)-G_2(t)\|\le\epsilon,\quad\forall t\in\mathbb R.
\end{equation}
Their time-evolution operators satisfy
\begin{equation}
\left\|\mathcal Te^{-i\int_0^tG_1(\tau)\,\mathrm d\tau}-\mathcal Te^{-i\int_0^tG_2(\tau)\,\mathrm d\tau}\right\|\le\epsilon t,\quad\forall t\ge0.
\end{equation}
\end{lemma}

This lemma implies that DTC is stable up to time $\ge\poly(N)$ with respect to perturbations of strength $1/\poly(N)$.

We say that a time-dependent perturbation $H_\textnormal{per}(t)$ is extensive if $H_\textnormal{per}(t)$ for any $t\in\mathbb R$ is extensive in the sense of Definition \ref{def:extp}.

\begin{conjecture} [\cite{vKS16, EBN16, KMS19, ZLM+23}] \label{c:2}
In the Floquet system (\ref{eq:modelb}), there is $(\mathbf h,\mathbf J)\in[-20,20]^{\times(3N-1)}$ at which DTC behavior is stable up to time $e^{\Omega(N)}$ with respect to extensive periodic perturbations of period $1$.
\end{conjecture}

Proving or disproving this conjecture is a remarkable achievement in understanding the stability of DTC at short and intermediate times.

\section*{Data availability statement}

No new data were created or analysed in this study.

\section*{Acknowledgments}

I would like to thank Jack Kemp and especially Norman Y. Yao for their suggestions which helped to improve the presentation of the paper. This work was supported by the NSF QLCI program (Grant No.~OMA-2016245).

\appendix

\section{Proofs} \label{app}

Consider a system of $N$ spins. Let $d_\textnormal{loc}$ be the local Hilbert space dimension of each spin. Since the Floquet operator (\ref{eq:flqo}) is unitary, it can be decomposed as
\begin{equation}
U_F=\sum_{j=1}^s\lambda_j\Pi_j,
\end{equation}
where $\lambda_1,\lambda_2,\ldots,\lambda_s$ are the distinct eigenvalues of $U_F$, and $\Pi_j=\Pi_j^2$ is the projector onto the eigenspace corresponding to $\lambda_j$. Let $d_j=\tr\Pi_j$ be the multiplicity of the eigenvalue $\lambda_j$ so that
\begin{equation}
\sum_{j=1}^sd_j=d_\textnormal{loc}^N.
\end{equation}

Let
\begin{equation}
D_1:=d_\textnormal{loc}^{-N}\sum_{j=1}^sd_j^2,\quad D_2:=\max_{w\in\mathbb C\setminus\{1\}}|\{(j,k):w=\lambda_k/\lambda_j\}|.
\end{equation}
It is easy to see that
\begin{equation}
1\le D_1\le d_\textnormal{loc}^N,\quad1\le D_2\le s.
\end{equation}
$D_1$ is a measure of how degenerate the spectrum of $U_F$ is. $D_1=1$ if and only if the spectrum is non-degenerate, i.e., $d_j=1$ for all $j$. Furthermore,
\begin{equation}
D_1\le\max_{j\in\{1,2,\ldots,s\}}d_j.
\end{equation}
Note that this bound can be very loose. For example, $D_1=O(1)$ if
\begin{equation}
d_j=
\begin{cases}
O(d_\textnormal{loc}^{N/3}),&j\le d_\textnormal{loc}^{N/3}\\
O(1),&j>d_\textnormal{loc}^{N/3}
\end{cases},
\end{equation}
where the multiplicities of an exponential number of eigenvalues are exponential in $N$.

The initial state can be expanded as
\begin{equation}
|\psi(0)\rangle=\sum_{j=1}^sc_j|j\rangle,\quad c_j:=\|\Pi_j|\psi(0)\rangle\|\ge0,\quad c_j|j\rangle:=\Pi_j|\psi(0)\rangle,\quad\sum_{j=1}^sc_j^2=1.
\end{equation}
The effective dimension of $|\psi(0)\rangle$ is defined as
\begin{equation}
D_\textnormal{eff}=1\big/\sum_{j=1}^sc_j^4.
\end{equation}
When $t$ is an integer,
\begin{equation} \label{eq:int}
|\psi(t)\rangle=\sum_{j=1}^sc_jU_F^t|j\rangle=\sum_{j=1}^sc_j\lambda_j^t|j\rangle.
\end{equation}

The following lemma is a generalization of previous results \cite{Tas98, Rei08, LPSW09, Sho11, SF12} on the equilibration of systems governed by a time-independent Hamiltonian to Floquet systems.

\begin{lemma} \label{l:eq}
For any observable $A$ with $\|A\|=1$, $\langle\psi(t)|A|\psi(t)\rangle$ is $\sqrt{D_2/D_\textnormal{eff}}$-periodic.
\end{lemma}

\begin{proof}
Let
\begin{equation}
\bar A(x):=\lim_{\mathbb N\ni M\to\infty}\frac1M\sum_{m=0}^{M-1}\langle\psi(m+x)|A|\psi(m+x)\rangle,\quad x\in[0,1).
\end{equation}
This function plays the role of $\bar f$ in Definition \ref{def:apf}. Let
\begin{equation}
\Delta A(x):=\lim_{\mathbb N\ni M\to\infty}\frac1M\sum_{m=0}^{M-1}\big|\langle\psi(m+x)|A|\psi(m+x)\rangle-\bar A(x)\big|^2,\quad x\in[0,1).
\end{equation}

Since $\lambda_1,\lambda_2,\ldots,\lambda_s$ are pairwise distinct complex numbers of absolute value $1$,
\begin{equation}
\lim_{\mathbb N\ni M\to\infty}\frac1M\sum_{m=0}^{M-1}\frac{\lambda_k^m}{\lambda_j^m}=\delta_{jk},
\end{equation}
where $\delta$ is the Kronecker delta. Let $A_{jk}:=\langle j|A|k\rangle$. Using Eq. (\ref{eq:int}),
\begin{equation}
\bar A(0)=\lim_{\mathbb N\ni M\to\infty}\frac1M\sum_{m=0}^{M-1}\sum_{j,k=1}^sc_jc_kA_{jk}\frac{\lambda_k^m}{\lambda_j^m}=\sum_{j=1}^sc_j^2A_{jj}=\tr(\psi_\textnormal{diag}A),
\end{equation}
where $\psi_\textnormal{diag}:=\sum_{j=1}^sc_j^2|j\rangle\langle j|$ is the Floquet analogue \cite{LDM14PRL, LDM14PRE} of the diagonal ensemble \cite{RDO08}. Since $(AA^\dag)_{jj}\le\|A\|^2=1$,
\begin{align}
&\Delta A(0)=\lim_{\mathbb N\ni M\to\infty}\frac1M\sum_{m=0}^{M-1}\left|\sum_{j\neq k}c_jc_kA_{jk}\frac{\lambda_k^m}{\lambda_j^m}\right|^2=\sum_{j\neq k;j'\neq k'}c_jc_kc_{j'}c_{k'}A_{jk}A_{j'k'}^*\lim_{\mathbb N\ni M\to\infty}\frac1M\sum_{t=0}^{M-1}\frac{\lambda_k^m\lambda_{j'}^m}{\lambda_j^m\lambda_{k'}^m}\nonumber\\
&=\sum_{j\neq k;j'\neq k'}c_jc_kc_{j'}c_{k'}A_{jk}A_{j'k'}^*\delta_{\frac{\lambda_k}{\lambda_j},\frac{\lambda_{k'}}{\lambda_{j'}}}\le\sum_{j\neq k;j'\neq k'}\frac{c_j^2c_k^2|A_{jk}|^2+c_{j'}^2c_{k'}^2|A_{j'k'}|^2}2\delta_{\frac{\lambda_k}{\lambda_j},\frac{\lambda_{k'}}{\lambda_{j'}}}\nonumber\\
&=\sum_{j\neq k}c_j^2c_k^2|A_{jk}|^2\sum_{j'\neq k'}\delta_{\frac{\lambda_k}{\lambda_j},\frac{\lambda_{k'}}{\lambda_{j'}}}\le D_2\sum_{j\neq k}c_j^2c_k^2|A_{jk}|^2\le D_2\sqrt{\sum_{j\neq k}c_j^4A_{jk}(A^\dag)_{kj}\times\sum_{j\neq k}c_k^4(A^\dag)_{kj}A_{jk}}\nonumber\\
&\le D_2\sqrt{\sum_{j=1}^sc_j^4(AA^\dag)_{jj}\times\sum_{k=1}^sc_k^4(A^\dag A)_{kk}}\le D_2\sum_{j=1}^sc_j^4=D_2/D_\textnormal{eff}.
\end{align}

Replacing $A$ in the above calculation by $B:=(U(0,x))^\dag AU(0,x)$, we obtain
\begin{equation} 
\Delta A(x)=\Delta B(0)\le D_2/D_\textnormal{eff},\quad\forall x\in[0,1)
\end{equation}
so that
\begin{equation} \label{eq:lop}
\lim_{\mathbb N\ni M\to\infty}\frac1M\sum_{m=0}^{M-1}\int_0^1\big|\langle\psi(m+x)|A|\psi(m+x)\rangle-\bar A(x)\big|^2\,\mathrm dx=\int_0^1\Delta A(x)\,\mathrm dx\le D_2/D_\textnormal{eff}.
\end{equation}
Finally, Markov's inequality implies that $\langle\psi(t)|A|\psi(t)\rangle$ is $\sqrt{D_2/D_\textnormal{eff}}$-periodic.
\end{proof}

The following lemma is a straightforward generalization of Lemma 5 in Ref.~\cite{HH19}.

\begin{lemma} \label{l:eff}
For a Haar-random product state $|\phi\rangle$ (Definition \ref{def:haar}),
\begin{equation}
\e_{|\phi\rangle}\sum_{j=1}^s(\langle\phi|\Pi_j|\phi\rangle)^2\le\frac{D_12^N}{(d_\textnormal{loc}+1)^N}.
\end{equation}
\end{lemma}

\begin{proof}
Recall that $|\phi\rangle=\bigotimes_{l=1}^N|\phi_l\rangle$, where each $|\phi_l\rangle$ is an independent Haar-random state.
\begin{lemma} [\cite{Har13}]
\begin{equation}
\e_{|\phi_l\rangle}\big((|\phi_l\rangle\langle\phi_l|)^{\otimes2}\big)=\frac{2\Pi_\textnormal{sym}}{d_\textnormal{loc}^2+d_\textnormal{loc}},
\end{equation}
where $\Pi_\textnormal{sym}=\Pi_\textnormal{sym}^2$ is the projector onto the symmetric subspace of two qudits.
\end{lemma}

Using this lemma,
\begin{multline}
\e_{|\phi\rangle}(\langle\phi|\Pi_j|\phi\rangle)^2=\e_{|\phi\rangle}\tr\big(\Pi_j^{\otimes2}(|\phi\rangle\langle\phi|)^{\otimes2}\big)=\tr\left(\Pi_j^{\otimes2}\e_{|\phi\rangle}\big((|\phi\rangle\langle\phi|)^{\otimes2}\big)\right)\\
=\tr\left(\Pi_j^{\otimes2}\bigotimes_{l=1}^N\e_{|\phi_l\rangle}\big((|\phi_l\rangle\langle\phi_l|)^{\otimes2}\big)\right)=\tr\left(\Pi_j^{\otimes2}\left(\frac{2\Pi_\textnormal{sym}}{d_\textnormal{loc}^2+d_\textnormal{loc}}\right)^{\otimes N}\right)\le\frac{d_j^22^N}{(d_\textnormal{loc}^2+d_\textnormal{loc})^N}.
\end{multline}
We complete the proof of Lemma \ref{l:eff} by summing over $j$.
\end{proof}

If $D_1$ is not too large, Markov's inequality implies that the effective dimension $D_\textnormal{eff}$ of a Haar-random initial product state $|\psi(0)\rangle$ is exponential in $N$ with overwhelming probability.

\begin{corollary} \label{c:t2}
If
\begin{equation}
D_1=e^{o(N)},
\end{equation}
then
\begin{equation}
\Pr_{|\psi(0)\rangle}\left(D_\textnormal{eff}>\left(\frac{d_\textnormal{loc}+1}2-10^{-6}\right)^N\right)=1-e^{-\Omega(N)}.
\end{equation}
If
\begin{equation}
D_1\le\left(\frac{d_\textnormal{loc}+1}2-10^{-6}\right)^N,
\end{equation}
then
\begin{equation}
\Pr_{|\psi(0)\rangle}(D_\textnormal{eff}=e^{\Omega(N)})=1-e^{-\Omega(N)}.
\end{equation}
\end{corollary}

We now study the spectral properties of Floquet systems. Intuitively, we expect that the set of Floquet systems with $D_1D_2>1$ has measure zero in the space of all Floquet systems. This statement is not mathematically precise and cannot be proved before ``measure zero'' is defined in the space.

We prove that in each ensemble of piecewise time-independent Floquet systems considered in Theorems \ref{t:b}, \ref{t:f}, the condition $D_1D_2=1$ is satisfied almost everywhere.

\begin{lemma} \label{l:ndg}
Consider the Floquet system (\ref{eq:f}). For any tuple $(n,T,\alpha,\gamma)$ such that $\alpha_1^x=\alpha_1^z=\gamma_1^{zz}=1$, the set
\begin{equation}
\{(\mathbf h,\mathbf J)\in\mathbb R^{12nN-9n}:D_1D_2>1\}
\end{equation}
has Lebesgue measure zero.
\end{lemma}

\begin{proof}
Recall that $\lambda_j$ is an eigenvalue with multiplicity $d_j$ of the Floquet operator
\begin{equation}
U_F=e^{-iH_n(t_n-t_{n-1})}e^{-iH_{n-1}(t_{n-1}-t_{n-2})}\cdots e^{-iH_1(t_1-t_0)}.
\end{equation}
Let
\begin{equation}
\mu_j:=\lambda_{j'}\quad\textnormal{for}\quad\sum_{k=1}^{j'-1}d_k<j\le\sum_{k=1}^{j'}d_k
\end{equation}
so that $\mu_1,\mu_2,\ldots,\mu_{2^N}$ are a complete set of eigenvalues of $U_F$. Since $U_F$ is unitary, one can write
\begin{equation}
\mu_j=e^{-iE_j},\quad E_j\in[-\pi,\pi),\quad\forall j,
\end{equation}
where $E_j$ is called a quasienergy of the Floquet system. The condition $D_1D_2=1$ is equivalent to the condition that the differences $\{(E_j-E_k)\bmod 2\pi\}_{j\neq k}$ between quasienergies with different indices are all distinct, i.e., for any $j\neq k$,
\begin{equation}
E_j-E_k\equiv E_{j'}-E_{k'}\pmod{2\pi}\quad\implies\quad(j=j')~\textnormal{and}~(k=k').
\end{equation}
This is reminiscent of the non-degenerate gap condition (\ref{eq:ndg}) for the spectrum of a time-independent Hamiltonian.

The real and imaginary parts of every element of the complex matrix $e^{-iH_j(t_j-t_{j-1})}$ are real analytic functions of $(\mathbf h,\mathbf J)$. So are the real and imaginary parts of the elements of $U_F$. Let
\begin{equation}
F:=\prod_{((j\neq j')\textnormal{ or }(k\neq k'))\textnormal{ and }((j\neq k')\textnormal{ or }(k\neq j'))}(\mu_j\mu_k-\mu_{j'}\mu_{k'})
\end{equation}
so that $F=0$ if and only if $D_1D_2=1$. By definition, $F$ is a symmetric polynomial in $\mu_1,\mu_2,\ldots,\mu_{2^N}$. The fundamental theorem of symmetric polynomials implies that $F$ can be expressed as a polynomial in $\Lambda_1,\Lambda_2,\ldots$, where
\begin{equation}
\Lambda_k:=\sum_{j=1}^{2^N}\mu_j^k=\tr(U_F^k).
\end{equation}
We see that every $\Lambda_k$ and hence $F$ are polynomials in the elements of $U_F$. Therefore, the real and imaginary parts of $F$ are real analytic functions of $(\mathbf h,\mathbf J)$. Since the zero set of a real analytic function has measure zero unless the function is identically zero \cite{Mit20}, it suffices to find a particular $(\mathbf h,\mathbf J)\in\mathbb R^{12nN-9n}$ such that $D_1D_2=1$. To construct such an example, let
\begin{equation}
h_{1,l}^x:=\epsilon g_l^x,\quad h_{1,l}^z:=\epsilon g_l^z,\quad J_{1,l}^{zz}:=\epsilon K_l,\quad\forall l,
\end{equation}
where $g_l^x,g_l^z,K_l$ are given in Lemma \ref{l} below. Let all other elements of $(\mathbf h,\mathbf J)$ be $0$. For sufficiently small $\epsilon>0$, $|E_j|<\pi/2$ for all $j$ so that
\begin{equation}
E_j-E_k\equiv E_{j'}-E_{k'}\pmod{2\pi}\quad\iff\quad E_j-E_k=E_{j'}-E_{k'}.
\end{equation}
Thus, the condition $D_1D_2=1$ follows from Lemma \ref{l}.
\end{proof}

\begin{lemma} [\cite{Hua21PP}] \label{l}
There exists $(g_l^u|_{1\le l\le N}^{u\in\{x,z\}},K_l|_{1\le l\le N-1})\in\mathbb R^{3N-1}$ such that the spectrum of the Hamiltonian
\begin{equation} \label{eq:ndgH}
\sum_{l=1}^N(g_l^x\sigma_l^x+g_l^z\sigma_l^z)+\sum_{l=1}^{N-1}K_l\sigma_l^z\sigma_{l+1}^z
\end{equation}
is non-degenerate and has non-degenerate gaps (Definition \ref{def:ndg}).
\end{lemma}

\begin{corollary} \label{c:b}
For the Floquet system (\ref{eq:modelb}), the set
\begin{equation}
\{(\mathbf h,\mathbf J)\in\mathbb R^{3N-1}:D_1D_2>1\}
\end{equation}
has Lebesgue measure zero.
\end{corollary}

\begin{proof}
As in the proof of Lemma \ref{l:ndg}, it suffices to find a particular $(\mathbf h,\mathbf J)\in\mathbb R^{3N-1}$ such that $D_1D_2=1$. To this end, let
\begin{equation}
h_l^z:=\epsilon g_l^z,\quad J_l:=\epsilon K_l,\quad h_l^x:=\epsilon g_l^x,\quad\forall l,
\end{equation}
where $g_l^x,g_l^z,K_l$ are given in Lemma \ref{l} above. Let $e_1<e_2<\cdots<e_{2^N}$ be the eigenvalues of (\ref{eq:ndgH}) and $E'_1\le E'_2\le\cdots\le E'_{2^N}$ be the quasienergies of the unitary operator $e^{-i(H'_1+H'_2)/2}$. For sufficiently small $\epsilon>0$, $E'_j=\epsilon e_j/2$ for all $j$. The non-degenerate gap condition (\ref{eq:ndg}) implies that
\begin{equation} \label{eq:ndge}
0<\min_{((j\neq j')\textnormal{ or }(k\neq k'))\textnormal{ and }((j\neq k')\textnormal{ or }(k\neq j'))}|E'_j+E'_k-E'_{j'}-E'_{k'}|\propto\epsilon.
\end{equation}
Let $E_1\le E_2\le\cdots\le E_{2^N}$ be the quasienergies of the Floquet operator $e^{-iH'_2/2}e^{-iH'_1/2}$. In the limit $\epsilon\to0^+$,
\begin{equation}
\|e^{-i(H'_1+H'_2)/2}-e^{-iH'_2/2}e^{-iH'_1/2}\|=O(\epsilon^2)
\end{equation}
so that \cite{BD84}
\begin{equation}
|E_j-E'_j|=O(\epsilon^2),\quad\forall j.
\end{equation}
This equation and (\ref{eq:ndge}) imply $D_1D_2=1$.
\end{proof}

\begin{lemma} \label{l:r}
$\psi(t)_S$ is $\sqrt[4]{d_S^2D_2/D_\textnormal{eff}}$-periodic, where $d_S:=d_\textnormal{loc}^L$ is the dimension of the Hilbert space of subsystem $S$.
\end{lemma}

\begin{proof}
Let
\begin{equation}
\bar\psi(x)_S:=\lim_{\mathbb N\ni M\to\infty}\frac1M\sum_{m=0}^{M-1}\psi(m+x)_S,\quad x\in[0,1).
\end{equation}
This function plays the role of $\bar f$ in Definition \ref{def:apmf}.

The rest of the proof largely follows the calculation in Section 5 of Ref.~\cite{Sho11}. Let $\{|0\rangle_S,|1\rangle_S,\ldots,|d_S-1\rangle_S\}$ be an arbitrary orthonormal basis of the Hilbert space of $S$. Define $d_S^2$ operators
\begin{equation}
A_{d_Sj_1+j_2+1}:=\frac1{\sqrt{d_S}}\sum_{k=0}^{d_S-1}e^{2\pi ij_2k/d_S}|(j_1+k)\bmod d_S\rangle_S\langle k|_S
\end{equation}
for $j_1,j_2=0,1,\ldots,d_S-1$. It is easy to see that
\begin{equation}
\|A_j\|=1/\sqrt{d_S},\quad\tr(A_j^\dag A_k)=\delta_{jk}.
\end{equation}
Since $\{A_j\}_{j=1}^{d_S^2}$ is a complete basis of operators on $S$, we expand
\begin{equation}
\psi(t)_S-\bar\psi(t-\lfloor t\rfloor)_S=\sum_{j=1}^{d_S^2}a_j(t)A_j,\quad a_j(t)=\tr\left(A_j^\dag\big(\psi(t)_S-\bar\psi(t-\lfloor t\rfloor)_S\big)\right),
\end{equation}
where $\lfloor\cdot\rfloor$ is the floor function. (\ref{eq:lop}) implies that
\begin{equation}
\lim_{M\to\infty}\frac1M\int_0^M|a_j(t)|^2\,\mathrm dt\le\frac{\|A_j^\dag\|^2D_2}{D_\textnormal{eff}}=\frac{D_2}{d_SD_\textnormal{eff}}.
\end{equation}
Using the relationship between the trace and Frobenius norms,
\begin{equation}
\|\psi(t)_S-\bar\psi(t-\lfloor t\rfloor)_S\|_1\le\sqrt{d_S\sum_{j,k=1}^{d_S^2}a_j^*(t)a_k(t)\tr(A_j^\dag A_k)}=\sqrt{d_S\sum_{j=1}^{d_S^2}|a_j(t)|^2}
\end{equation}
so that
\begin{equation}
\lim_{M\to\infty}\frac1M\int_0^M\|\psi(t)_S-\bar\psi(t-\lfloor t\rfloor)_S\|_1\,\mathrm dt\le\sqrt{\sum_{j=1}^{d_S^2}\lim_{M\to\infty}\frac{d_S}M\int_0^M|a_j(t)|^2\,\mathrm dt}\le d_S\sqrt{\frac{D_2}{D_\textnormal{eff}}}.
\end{equation}
Finally, Markov's inequality implies that $\psi(t)_S$ is $\sqrt[4]{d_S^2D_2/D_\textnormal{eff}}$-periodic.
\end{proof}

Theorem \ref{t:b} follows from Lemma \ref{l:eq} and Corollaries \ref{c:t2}, \ref{c:b}. Theorem \ref{t:f} follows from Lemmas \ref{l:eq}, \ref{l:ndg} and Corollary \ref{c:t2}. Theorem \ref{t:r} follows from Lemma \ref{l:r} because
\begin{equation}
\sqrt[4]{d_S^2D_2/D_\textnormal{eff}}=e^{-\Omega(N)}\impliedby
\begin{cases}
d_S=2^L\le2^{0.29248N}&\text{condition in Theorem \ref{t:r}}\\
D_2=1&\text{Lemma \ref{l:ndg}}\\
D_\textnormal{eff}\ge(3/2-10^{-6})^N&\text{Corollary \ref{c:t2}}
\end{cases}.
\end{equation}

\printbibliography

\end{document}